\documentclass[a4paper,conference]{IEEEtran}
\usepackage[english]{babel}
\usepackage{amsmath,amsthm}
\usepackage{amsfonts}
\usepackage{extarrows}
\usepackage{amssymb}
\usepackage{dblfloatfix}
\twocolumn
\newtheorem{theorem}{Theorem}
\newtheorem{corollary}{Corollary}
\newtheorem{lemma}{Lemma}
\newtheorem{proposition}{Proposition}
\theoremstyle{definition}

\theoremstyle{remark}
\newtheorem{remark}{Remark}

\newcommand{\F}{\mathbb{F}}
\newcommand{\e}{\varepsilon}
\newcommand{\bb}{\F_2}

\begin{document}\IEEEoverridecommandlockouts
\title{Polynomial complexity of polar codes for non-binary alphabets, key agreement and Slepian-Wolf coding}
\author{\IEEEauthorblockN{Jingbo Liu~~~~~~~~Emmanuel Abbe}
\thanks{Jingbo's work was supported in part by the National Science Foundation under Grants CCF-1116013 and the Air Force Office of Scientific Research under Grant FA9550-12-1-0196. Emmanuel's work was supported by NSF grant CIF-1706648. Parts of this work were presented in CISS 2014.}
\IEEEauthorblockA{Dept. of Electrical Eng., Princeton University, NJ 08544\\
\{jingbo,eabbe\}@princeton.edu}}
\maketitle
\begin{abstract}
We consider polar codes for memoryless sources with side information and show that the blocklength, construction, encoding and decoding complexities are bounded by a polynomial of the reciprocal of the gap between the compression rate and the conditional entropy. This extends the recent results of Guruswami and Xia to a slightly more general setting, which in turn can be applied to (1) sources with non-binary alphabets, (2) key generation for discrete and Gaussian sources, and (3) Slepian-Wolf coding and multiple accessing. In each of these cases, the complexity scaling with respect to the number of users is also controlled. In particular, we construct coding schemes for these multi-user information theory problems which achieve optimal rates with an overall polynomial complexity.
\end{abstract}
\IEEEpeerreviewmaketitle
\section{Introduction}
The original paper of Arikan \cite{ari1} introduces polar codes for binary input memoryless channels. It is shown that the encoding and decoding complexity is $O(n \log n)$, and while the code construction is a priori of exponential complexity, it is shown how it can be approximated using simulations. In \cite{tal2011construct}, the code construction is studied using lower and upper bounds on the polarized mutual informations which are efficiently computable. More recently, \cite{guruswami2013polar} with related developments in \cite{hassani2013polarization} provides a complete and refined analysis of these techniques to obtain a global complexity control for the block length, construction, encoding, and decoding complexity which is polynomial in the block length and in the reciprocal of the gap to capacity.

In a subsequent paper to \cite{ari1}, Arikan introduces polar codes for memoryless sources with side-information \cite{arikan2010source}. It is interesting to compare the generality of this paper with respect to the paper on channel polarization \cite{ari1}. While channel polarization does imply source polarization when the source does not have side-information (by taking an additive noise channel and using source-channel duality, see for example \cite{dublin}), channel polarization does not imply the source polarization with side-information. In particular, the latter setting requires an extension of the channel setting which requires a uniform input distribution, whereas the source setting with side-information does not have such as restriction. In this note, we fill in this gap and show that in the general setting of sources with side information, the results of Guruswami et Xia can be extended: within $\epsilon$ gap to the conditional entropy, there exist source polar codes whose block length/construction/encoding/decoding complexity are bounded by polynomials in $1/\epsilon$. This generalization is not difficult but is particularly interesting as it opens immediately to several other extensions: (i) the results are extended to the case where the source cardinality is a power of 2, suggested as future work in \cite{guruswami2013polar} (ii) the results are extended to a secret key generation setting \cite{koyluoglu2010polar,chou2013polar}, (iii) the results are extended to Slepian-Wolf coding \cite{arikan2010source} and multiple accessing \cite{csacsouglu2011polarization} using onion-peeling decoding.

In particular, for alphabets of size $2^m$, Slepian-Wolf coding and multiple accessing problems with $m$ users, it is shown that complexity scaling is also linear in the number of users. This an interesting feature compared to the schemes developed in \cite{abbe2012polar,polar_corr} for the same settings which rely on the joint decoding of the users, for which linear complexity in the number of users is not achieved. On the other hand, the onion-peeling approach only guarantees rates on the corner-point of the capacity regions, and requires time-sharing for other rates on the dominant face. Concerning the secret key agreement, we consider in this note only the case of a uniform marginal distribution for Alice, and leave the general setup of \cite{chou2013polar} for future work.
We show in addition that the method extends to correlated Gaussian sources, using an approximation method similar to (but not exactly the same as) the one used in the polar coding for the AWGN channel~\cite{abbe2011polar}.

\section{Source Polar Code Construction}
In this section we provide a source coding counterpart of \cite{guruswami2013polar} with side information. Although the main ingredients also exist in the literature in the source coding setting, we shall synthesize these results to show that with polar coding, the source coding block length can be bounded by a polynomial of the gap between compression rate and conditional entropy, while ensuring tractable encoding/decoding complexity.

Following \cite{guruswami2013polar}, we use the terms ``\emph{rough polarization}'', ``\emph{fine polarization}'' and ``\emph{degradation}'' for the three main ingredients in the polar code construction.
\subsection{Evolution of Source Bhattacharyya coefficient}
For correlated random variables $X,Y$ where $X\in\bb$, define the source Bhattacharyya coefficient as in \cite{arikan2010source}
\begin{align}
Z(X|Y):=&\sum_{y\in \mathcal{Y}}P_Y(y)\sqrt{P_{X|Y}(0|y)P_{X|Y}(1|y)}\\
=&\sum_{y\in \mathcal{Y}}\sqrt{P_{XY}(0,y)P_{XY}(1,y)}.
\end{align}
\begin{lemma}\label{lem4}
Suppose $(X_i,Y_i)$ are i.i.d. according to $P_{XY}$. We have
\begin{align}
Z(X_2|Y^2,X_1+X_2)=&Z(X|Y)^2;\label{ev1}\\
Z(X|Y)\sqrt{2-Z(X|Y)^2}\le &Z(X_1+X_2|Y^2)\label{ev2_1}\\
\le& 2Z(X|Y)-Z(X|Y)^2.\label{ev2}
\end{align}
\end{lemma}
\begin{proof}
Equality (\ref{ev1}) and the second equality in (\ref{ev2}) are known; see for example \cite{arikan2010source}. We now prove the first inequality in (\ref{ev2}) with a similar technique as in \cite{guruswami2013polar}. Using the basic definition and upon rearranging, we find
\begin{align}
&Z(X_1+X_2|Y^2)\nonumber\\
=&2\sum_{y^2}\sqrt{P_{X_1+X_2,Y^2}(0,y^2)P_{X_1+X_2,Y^2}(0,y^2)}\\
=&2\sum_{y^2}\sqrt{P_{X_1Y_1}(0,y_1)P_{X_1Y_1}(1,y_1)P_{X_2Y_2}(1,y_2)P_{X_2Y_2}(0,y_2)}\nonumber\\
&\cdot\sqrt{f(y_1)^2+f(y_2)^2-4},
\end{align}
where we have defined $f(y):=\sqrt{\frac{P_{XY}(0,y)}{P_{XY}(1,y)}}+\sqrt{\frac{P_{XY}(1,y)}{P_{XY}(0,y)}}$.
Now let $p(y):=\frac{2}{Z(X|Y)}\sqrt{P_{XY}(0,y)P_{XY}(1,y)}$; by applying Jensen's inequality twice,
\begin{align}
&Z(X_1+X_2|Y^2)\nonumber\\
=&\frac{Z(X|Y)^2}{2}\mathbb{E}_{y_1,y_2\sim p(y)}\sqrt{f(y_1)^2+f(y_2)^2-4}\\
\ge&\frac{Z(X|Y)^2}{2}\sqrt{(\mathbb{E}_{y_1\sim p(y)}f(y_1))^2+(\mathbb{E}_{y_2\sim p(y)}f(y_2))^2-4}\\
\ge&Z(X|Y)\sqrt{2-Z(X|Y)^2}.
\end{align}
\end{proof}
Denote by $S$ the joint distribution of $X,Y$. Let $(X_1,Y_1),(X_2,Y_2)$ be i.i.d. according to $P_{XY}$. Define $S^+$ (resp. $S^-$) as the joint distribution of $X_2,(Y^2,X_1+X_2)$ (resp. $X_1+X_2,Y^2$). By Lemma \ref{lem4}, the evolution of source Bhattacharyya coefficients is similar to that of channel Bhattacharyya coefficients studied in \cite{guruswami2013polar}. Using these `$+$' and `$-$' operations, for $0\le i\le 2^n-1,N:=2^n$ we can define recursively the sequence of distributions $S^{(i)}_n:\bb\times(\mathcal{Y}^N\times\bb^i)$ via
\begin{align}\label{eq3}
S^{(i)}_{n+1}=\left\{\begin{array}{cc}
                      (S^{(\lfloor i/2\rfloor)}_n)^- & \textrm{$i$ is even} \\
                      (S^{(\lfloor i/2\rfloor)}_n)^+ & \textrm{$i$ is odd}
                    \end{array}
\right.
\end{align}
with the base distribution $S^{(0)}_0=S$. We use the shorthand notation $Z(S^{(i)}_n)$ for $Z(X|\bar{Y})$ where $\mathcal{X}=\bb$, $\mathcal{\bar{Y}}=\mathcal{Y}^N\times\bb^i$ and $(X,\bar{Y})$ is distributed according to $S^{(i)}_n$.
\subsection{Rough Polarization}
The name ``rough polarization'' is from the fact that the rate of polarization in this stage is not as fast as the fine polarization to be discussed later. The following result characterizes the speed of convergence for a type of supermartingale, which turns out to be very useful in the proof of rough polarization:
\begin{lemma}\cite[Lemma 7]{guruswami2013polar}\label{lem2}
Suppose $B_0,B_1,\dots$ is a sequence of i.i.d. $Ber(0.5)$ random variables. A supermartingale with respect to the filtration $\sigma(B_0^n)$ satisfies
\begin{align}
p_0&=p\in(0,1),\\
p_n&=p_{n-1}^2,\textrm{if $B_n=1$},\\
p_{n-1}\sqrt{2-p_{n-1}^2}\le p_n&\le 2p_{n-1}-p_{n-1}^2,\textrm{if $B_n=0$}.\label{e3}
\end{align}
Then $\mathbb{E}[(p_n(1-p_n))^{1/2}]\le \frac{1}{2}\Lambda^n$ for some $0<\Lambda<1$.
\end{lemma}
\begin{remark}
In the case of erasure channel, (\ref{e3}) can be replaced with the exact formula $p_n=2p_{n-1}-p_{n-1}^2$, and hence we can improve the result with $\mathbb{E}[q_n^{1/2}]\le \frac{1}{2}\left(\frac{3}{4}\right)^{n/2}$, see \cite{tanaka2010properties}.
\end{remark}
\begin{remark}
For an arbitrary channel, the value of $\Lambda$ can be as small as $1.85/2$ \cite{hassani2010scaling}.
\end{remark}
The proof of rough polarization of channel Bhattacharyya coefficients in the literature is essentially based on Lemma \ref{lem2}.
Now, by Lemma \ref{lem4}, the evolution of the source Bhattacharyya coefficients can also be thought of as the type of supermartingale considered in Lemma \ref{lem2}. Thus we obtain the following result about rough polarization of source Bhattacharyya coefficients. The proof is omitted since it is similar to the proof of rough polarization of channel Bhattacharyya coefficients (c.f. \cite[Proposition 5]{guruswami2013polar}).
\begin{proposition}
For joint distribution $p_{XY}$ with $\mathcal{X}=\bb$ and $\rho\in(\Lambda^2,1)$ (where $\Lambda$ is as in Lemma \ref{lem2}), there is a constant $b_{\rho}$ which only depends on $\rho$ such that for all $0<\epsilon<\frac{1}{2}$, and $m\ge b_{\rho}\log(1/\epsilon)$, there exists a roughly polarized set
\begin{align}
\mathcal{S}_r\subset\mathcal{S}:=\{S^{(i)}_m:0\le i\le 2^m-1\}
\end{align}
such that for all $M\in\mathcal{S}_r$, $Z(M)\le 2\rho^m$ and ${\rm Pr}_i(S^{(i)}_m\in \mathcal{S}_r)\ge I(X;Y)-\epsilon$.
\end{proposition}

\subsection{Fine Polarization}
The rough polarization stage produces a set of size nearly $I(X;Y)N$ in which the source Bhattacharyya coefficients are moderately small. They are not small enough to show vanishing probability of decoding error. However, they are small enough such that just by tracking the upper bounds in (\ref{ev1}) and (\ref{ev2}) (which corresponds to the so called extremal process), we can determine a large fraction of very small Bhattacharyya coefficients originated from that set as the branching process goes on. This idea is originally proposed in \cite{arikan2009rate}.

Since the fine polarization stage only depends on the extremal process, there is not much new work to be done to obtain a source coding counterpart of \cite[Proposition 10]{guruswami2013polar}. The following fixes a small error in the proof of \cite[Proposition 10]{guruswami2013polar}:

\begin{lemma}
Given $\gamma>0$, $\beta\in(0,\frac{1}{2})$ and $\rho\in (0,1)$, there is a constant $\theta(\beta,\gamma,\rho)$ such that for all $0<\epsilon<\frac{1}{2}$, if $m>\theta(\beta,\gamma,\rho)\cdot \log(2/\epsilon)$ then
\begin{align}\label{elem}
\frac{4\gamma}{\lg(1/\rho)}\exp\left(-\frac{(1-2\beta)^2m\lg(1/\rho)}{8}\right)<\frac{\epsilon}{2}.
\end{align}
\end{lemma}
\begin{proof}
Viewing (\ref{elem}) in the form $c_1\exp(-c_2 m)<\epsilon$, where $c_1=\frac{8\gamma}{\lg(1/\rho)}$, $c_2=\frac{(1-2\beta)^2\lg(1/\rho)}{8}$, we see it suffices to set $
\theta=\max\left\{\frac{\log(2c_1)}{c_2\log4},\frac{1}{c_2}\right\}.
$
\end{proof}
Accordingly, we can define $c_{\rho}=\lceil\frac{4n}{m\lg(2/\rho)}\rceil$ in the proof of Proposition 10 in \cite{guruswami2013polar}. Then, the step above equation (22) in their paper can be replaced by $\lg Z(M_n^{(i)})\le -\frac{m 2^{n\beta}}{5}\lg(1/\rho)$. We then obtain a source version of fine polarization.
\begin{proposition}
Given $\epsilon\in (0,\frac{1}{2})$, a joint distribution $P_{XY}$ with $\mathcal{X}=\bb$, a parameter $\delta\in(0,\frac{1}{2})$, there exists a constant $c_{\delta}$ such that if $n_0>c_{\delta}\log(1/\epsilon)$ then
\begin{align}
{\rm Pr}_i[Z(S_{n_0}^{(i)})\le 2^{-2^{\delta n_0}}]\ge I(X;Y)-\epsilon.
\end{align}
\end{proposition}

\subsection{Efficient Construction using Degradation}
From (\ref{eq3}), $S^{(i)}_n$ can be seen as a distribution on the set $\bb\times\mathcal{\bar{Y}}$, where $\mathcal{\bar{Y}}:=(\mathcal{Y}^N\times\bb^i)$. Since $\mathcal{\bar{Y}}$ may have a large cardinality, the construction of polar codes is not efficient if we have to exactly compute the Bhattacharyya coefficients from (\ref{eq3}). The ``binning'' or ``degradation'' method, originally proposed in \cite{tal2011construct}, is designed to overcome this computational barrier. The idea is to find $T$ such that $T$ is almost a sufficient statistic of $\bar{Y}$ for $X$, but $|\mathcal{T}|$ is much smaller than $|\mathcal{\bar{Y}}|$. The degradation method can be performed after each branching process (\ref{eq3}) in the rough polarization stage. There is no need to use degradation in the fine polarization stage since the distribution is no longer involved in that stage.

Suppose $P_{X\bar{Y}}$ is a joint distribution on $\mathcal{X}\times\mathcal{\bar{Y}}$ where $\mathcal{X}=\bb$. (We use $\bar{Y}$ to indicate that it is not the same as the side information $Y$ we defined earlier.) Partition $\mathcal{\bar{Y}}$ into sets $\mathcal{\bar{Y}}_{i,j}$, $i=1,\dots,k$, $j=0,1$ and $\mathcal{\bar{Y}}_{k+1}$ defined as
\begin{align}
\mathcal{\bar{Y}}_{i,j}=&\{y:P_{X|\bar{Y}=y}(j)>P_{X|\bar{Y}=y}(j+1),\nonumber\\
&\frac{i-1}{k}\le H(P_{X|\bar{Y}=y})<\frac{i}{k}\},\quad i=1,\dots,k,\label{eq26}\\
\mathcal{\bar{Y}}_{k+1}=&\{y:P_{X|\bar{Y}=y}(0)=P_{X|\bar{Y}=y}(1)\}\label{eq261}.
\end{align}
Let $T$ be a r.v. taking values in $\{1,\dots,k\}\times\bb\cup\{k+1\}$ such that $X-\bar{Y}-T$ and for all $y:P_{\bar{Y}}(y)>0$ we have $P_{T|\bar{Y}}(t|y)=1_{y\in\mathcal{\bar{Y}}_t}$.
Then using the same method as \cite{pedarsani2011construction}, we can show that
\begin{proposition}\label{prop5}
If $\widehat{S^{(i)}_n}$ is the resulting distribution when each `$+$' or `$-$' operation is followed by a degradation step, then
\begin{align}
H(S^{(i)}_n)\le H(\widehat{S^{(i)}_n})\le H(S^{(i)}_n)+\frac{n2^n}{k},
\end{align}
where $k$ is as in (\ref{eq26}), (\ref{eq261}) so that the number of bins is $2k+1$.
\end{proposition}
Note that in \cite{guruswami2013polar}, the quantization is uniform in the space of $p\in[0,1]$. Here we are quantizing in the space of $h(p)\in [0,1]$, which will yield a slightly better result and cleaner analysis.

Combining the rough polarization, fine polarization and degradation together, we have the following main result which links complexity with the gap to entropy:
\begin{theorem}\label{thm1}
There is a constant $0<\mu<\infty$ such that the following holds: let $P_{XY}$ be a joint distribution with $\mathcal{X}=\bb$. There exists $a_S<\infty$ such that for all $0<\epsilon<\frac{1}{2}$ and powers of two $N\ge a_S(1/\epsilon)^{\mu}$, there is a source polar code of block length $N$ and rate below $H(X|Y)+\epsilon$ with construction time complexity ${\rm poly}(N)$. The encoding and decoding algorithms have time complexity $O(N\log N)$ and the error probability is at most $2^{-N^{0.49}}$.
\end{theorem}
\begin{remark}
In this theorem the constant $\mu$ is independent of the channel, whereas $a_S$ depends on the particular channel.
\end{remark}

\section{Extension to Non-binary Alphabets}\label{sec1}

From Theorem \ref{thm1}, one can design an `onion peeling' encoding scheme for sources with alphabet size of $2^m$, using the technique of polar coding for $m$-user MAC introduced in \cite{csacsouglu2011polarization,abbe2012polar}. The idea is to identify $X$ with its binary expansion $(X^{(1)},\dots,X^{(m)})$, where $|\mathcal{X}^{(i)}|=2$. Consider the expansion
\begin{align}\label{chain}
H(X|Y)=&H(X^{(1)}|Y)\nonumber\\
&+\dots+H(X^{(m)}|Y,X^{(1)},\dots,X^{(m-1)}).
\end{align}
If we encode and decode the $i$'th layer $(X_1^{(i)},\dots,X_N^{(i)})$ in the order $i=1,\dots,m$, then by Theorem~\ref{thm1} with a union bound ensures a low probability of incorrect decoding. The encoding rate will also be close to $H(X|Y)$ because of (\ref{chain}). More precisely, we have
\begin{corollary}\label{cor1}
There is a constant $0<\mu<\infty$ such that the following holds: let $P_{XY}$ be a joint distribution with $\mathcal{X}=\bb^m$. There exists $a_S<\infty$ such that for all $0<\epsilon<\frac{1}{2}$ and powers of two $N\ge a_S(1/\epsilon)^{\mu}$, there is a source polar code of block length $N$ and rate below $H(X|Y)+m\epsilon$ with construction time complexity $m{\rm poly}(N)$. The encoding and decoding algorithms have time complexity $O(mN\log N)$ and the error probability is at most $m2^{-N^{0.49}}$.
\end{corollary}
\begin{remark}
Since any discrete random variable can have its support embedded in a set of size $2^m$ for a large enough $m$, we can use the scheme in Corollary \ref{cor1} to compress arbitrary discrete memoryless sources.
\end{remark}
As we shall see in the next section, Corollary~\ref{cor1} can be applied to key generation from general sources after a quantization step.

\section{Application to Key Agreement}
Suppose terminals A,B observe discrete memoryless sources $X_i,Y_i$ respectively, where $X_i,Y_i$ are distributed according to $P_{XY}$. A public message $W=W(X_1^N)$ can be computed at terminal A and sent to terminal B. Then terminal A, B compute their secret keys $K=K(X_1^N)$ and $\hat{K}=\hat{K}(Y_1^N,W)$, respectively. The key rate is defined as
\begin{align}
R=\log|\mathcal{K}|
\end{align}
and we say perfect secrecy is achieved if
\begin{align}
R=H(K|W).
\end{align}

When unlimited public communication from A to B is allowed, it is well known that the key capacity is $I(X;Y)$. In the case where $\mathcal{X}$ is binary, practical key agreement schemes based on polar codes have been proposed: we can apply the efficient code construction in the previous section to the scheme described in \cite{chou2013polar} to obtain the performance guarantee of polar key generation algorithm.
\subsection{Equiprobable Case}
\begin{corollary}\label{cor2}
There is a constant $0<\mu<\infty$ such that the following holds: let $P_{XY}$ be a joint distribution of the sources observed at two terminals, where $P_X$ is the equiprobable distribution on $\bb^m$. There exists $a_S<\infty$ such that for all $0<\epsilon<\frac{1}{2}$ and powers of two $N\ge a_S(1/\epsilon)^{\mu}$, there is a key generation scheme such that the public message has  block length $N$ and rate below $H(X|Y)+m\epsilon$ with construction time complexity $m{\rm poly}(N)$; the key has rate above $I(X;Y)-m\epsilon$ and the encoding and decoding algorithms have time complexity $O(mN\log N)$. Moreover, the probability of $K\neq \hat{K}$ is at most $m2^{-N^{0.49}}$ and perfect secrecy is achieved.
\end{corollary}
\begin{proof}
The coding scheme is similar to \cite[Proposition 4.2]{chou2013polar} except that now $|\mathcal{X}|=2^m$ and the performance of the polar codes is guaranteed by Corollary \ref{cor1}.

As in \ref{sec1}, we identify $X$ with $(X^{(1)},\dots,X^{(m)})$. Define $(U^{(i)})^N=G_N (X^{(i)})^N$ for $i=1\dots m$, where
\begin{align}
G_N:=\left(
                              \begin{array}{cc}
                                1 & 1 \\
                                0 & 1 \\
                              \end{array}
                            \right)^{\bigotimes n},
\end{align}
and recall that $N=2^n$.
Define the sets
\begin{align}
\mathcal{F}^{(i)}:=&\{1\le j\le n:\nonumber\\
&~~H(U^{(i)}_j|[U^{(i)}]^{j-1},[U^{(i-1)}]^N,\dots,[U^{(1)}]^N,Y^N)\nonumber\\
&~~~~<2^{-N^{0.499}}\}.
\end{align}
the purpose of setting the threshold at $2^{-N^{0.499}}$ in the above is merely that $0.499>0.49$. Now we invoke Corollary \ref{cor1} (and its proof method), and assume that $\mu$ is as in Corollary \ref{cor1}. For each $1\le i\le m$ there exists $a_S^{(i)}<\infty$ such that for all $0<\epsilon<\frac{1}{2}$ and powers of two $N\ge a_S^{(i)}(1/\epsilon)^{\mu}$, we have
\begin{align}
\frac{|[\mathcal{F}^{(i)}]^c|}{N}\le H(X^{(i)}|X^{(i-1)},\dots,X^{(1)},Y)+\epsilon.
\end{align}
Then by chain rule,
\begin{align}
\frac{|\bigcup_{i=1}^m[\mathcal{F}^{(i)}]^c|}{N}\le H(X|Y)+k\epsilon,
\end{align}
and hence
\begin{align}
\frac{|\bigcup_{i=1}^m[\mathcal{F}^{(i)}]|}{N}\ge& m-(H(X|Y)+k\epsilon)\\
=&I(X;Y)-k\epsilon.
\end{align}
if terminal A sends $U^{(i)}_{[\mathcal{F}^{(i)}]^c}$, $i=1,\dots, m$ to terminal B, then B can decode $U^{(i)}_{[\mathcal{F}^{(i)}]}$, $i=1,\dots, m$ with error probability not exceeding $2^{-N^{0.49}}$. Therefore we can use $U^{(i)}_{[\mathcal{F}^{(i)}]}$ as the key bits. Perfect secrecy is achieved because $U^{(i)}_{[\mathcal{F}^{(i)}]^c}$ and $U^{(i)}_{[\mathcal{F}^{(i)}]}$ are independent. We can set $a_S=\max_{1\le i\le m}{a^{(i)}_S}$ so that the asserted block length can be achieved. The asserted encoding and decoding complexities are guaranteed by Corollary \ref{cor1}.
\end{proof}
In the following we shall discuss how to extend the method to the case where $\mathcal{X}$ is non-binary.

\subsection{Extensions}\label{subsub1}
If $P_X$ is not an equiprobable distribution on a set of size $2^m$, the key generation scheme in Corollary~\ref{cor2} does not work directly. In this case, we can consider the following trick: produce a degraded version $\tilde{X}$ of $X$ at terminal A. This means that $\tilde{X}-X-Y$. If $\tilde{X}$ is equiprobably distributed on an alphabet of size $2^m$, then we can apply the polar coding scheme in Corollary~\ref{cor1} to the new sources $\tilde{X},Y$, achieving a key capacity of $I(\tilde{X};Y)$. If we can choose $m$ large so that $I(\tilde{X};Y)\approx I(X;Y)$, then the key rate can approach the key capacity.

For application purposes it usually suffices to consider $\tilde{X}$ to be a quantization function of $X$. A similar trick has been used to approximate the capacity of non-symmetric channels using polar codes, c.f. \cite[Section 4.3]{csacsouglu2011polarization}.

As a prominent example, we shall analyze how the above trick can be applied to the problem of key generation from correlated Gaussian sources. The quantization method used here is reminiscent of, but actually different from, the quantization method for approaching capacity of AWGN channel using polar codes discussed in \cite{abbe2011polar}.

Assume that scalar r.v.'s $X,Y$ are jointly Gaussian with correlation coefficient $\rho$. The key capacity becomes $I(X;Y)=\frac{1}{2}\log\frac{1}{1-\rho^2}$. To approach the key capacity, one can find $\tilde{X}$ such that $\tilde{X}-X-Y$ and $\tilde{X}$ is equiprobably distributed on an alphabet of size $2^m$, and then use the key generation scheme in Corollary \ref{cor2}. The following result shows that for large $m$, one can ensure that the gap between $I(Y;\tilde{X})$ and $I(Y;X)$ is of the order of $\sqrt{\frac{m}{2^m}}$.
\begin{lemma}\label{lem7}
If $X$, $Y$ are jointly Gaussian with correlation coefficient $\rho$, then for large $k$, there exists $\tilde{X}$ which is a function of $X$ and equiprobably distributed on a set of size $k$, such that
\begin{align}
I(Y;\tilde{X})\ge I(Y;X)-\frac{\log e}{2}\cdot\frac{\rho^2 C}{1-\rho^2}\sqrt{\frac{\ln k}{k}}
\end{align}
for some $C>0$. Moreover, it suffices to choose $C>\sqrt{\frac{2}{\pi}}+2\sqrt{2\pi}$.
\end{lemma}
\begin{proof}
See Appendix \ref{prooflm7}.
\end{proof}
Instead of the scalar case, if $\bf X,Y$ are vector Gaussian random variables of dimension $d$, we can always find non-degenerate linear transforms $\bf X\mapsto \bar{X}$, $\bf Y\mapsto \bar{Y}$ such that $(\bar{X}_i,\bar{Y}_i)$ are i.i.d. pairs for $i=1,\dots,d$. Then the key capacity can be achieved using the optimal strategies in the scalar case; see \cite{jingbo} for details and generalizations.

\section{Slepian-Wolf coding and multiple accessing}
Source compression with side information can readily be applied to the Slepian-Wolf coding problem, as in \cite{arikan2010source}.


The Slepian-Wolf coding problem consists in compressing correlated sources without the encoders cooperating (after the code agreement).
Let $X_1,\dots,X_n$ be i.i.d.\ under $\mu$ on $\F_2^m$, i.e., $X_i$ is an $m$ dimensional binary random vector and $X_1[i],\dots,X_n[i]$ is the sources output for user $i$.
Compressing these sources by having access to all the realizations requires roughly $nH(\mu)$ bits.
In \cite{slepian}, Slepian and Wolf showed that, even if the encoders are not able to cooperate after observing the source realizations, lossless compression can still be achieved at sum rate $H(\mu)$.

A simple way to achieve this goal is via the ``onion-peeling'' approach. Each user $i \in [m]$ computes $U^n[i]=X^n[i] G_n$ and transmits to the central decoder the non-deterministic bits of $U^n[i]$ {\it conditioned} on the previous $i-1$ source sequences:
\begin{align}
\{j \in [n]: H(U_j[i] | U^{j-1}[i], U^n[1],\dots,U^n[i-1]) \geq \e \},
\end{align}
the central decoder can then successively decode each user, replacing the previous sequences by their estimate. The sum-rate of this code approaches $$\sum_{i \in [m]} H(X[i]|X[i-1])=H(X[1],\dots,X[m]).$$ Note however that with this approach, each user is operating at a corner-point of the rate region.
Using Theorem \ref{thm1} and standard arguments to control the error propagation, the following is obtained.
\begin{corollary}\label{coro_sw}
There is a constant $0<\mu<\infty$ such that the following holds: let $P_{X[1],\dots,X[m]}$ be a joint distribution on $\bb^n$. There exists $a_S<\infty$ such that for all $0<\epsilon<\frac{1}{2}$ and powers of two $N\ge a_S(1/\epsilon)^{\mu}$, there is a polar code of block length $N$ and sum-rate below $H(X[1],\dots,X[m]|Y)+m\epsilon$ with construction time complexity $m{\rm poly}(N)$. The encoding and decoding algorithms have time complexity $O(mN\log N)$ and the error probability is at most $m2^{-N^{0.49}}$.
\end{corollary}

Using duality arguments, a similar result can be obtained for the multiple access channels, achieving rates on the corner point of the capacity region.


\section{Conclusion and Future Work}
We have studied an efficient construction of polar codes for losslessly compressing a source $X$ with side information $Y$ at the decoder, where $|\mathcal{X}|$ is a power of two. It is shown that within $\epsilon$ gap to the conditional entropy, there exist source polar codes whose block length/construction/encoding/decoding complexity are bounded by polynomials in $1/\epsilon$, extending the realm of \cite{guruswami2013polar}. The key observation is that, as in the channel setting, the bounds (\ref{ev1})-(\ref{ev2}) still holds in the source setting, even though now $X$ is not necessarily an equiprobable distribution.

Future work may include applying the efficient source polar coding techniques to other secret key generation problems, such as key generation with limited public communication, or key generation under an eavesdropper's observation. It's also worthwhile to extend theorem~\ref{thm1} to prime alphabets. The main difficulty in such an extension is the lack of a prime alphabet counterpart of the lower bound in \eqref{ev2_1} for one step evolution of Bhattacharyya coefficient. It is possible to replace the analysis based on Bhattacharyya coefficient in the rough polarization stage with an analysis based on entropy \cite{guruswami2013polarnew}. Although an inequality regarding one step evolution of the entropy is known for prime alphabets \cite{sasoglu2010entropy}, it is not strong enough to be applied to Theorem~\ref{thm1}. More precisely, Theorem~\ref{thm1} requires an inequality in the form of Theorem~2 in \cite{sasoglu2010entropy} with $\epsilon(\delta)\gtrsim\delta$, which is not guaranteed from the proof technique of \cite{sasoglu2010entropy}. Another interesting direction is to pursue the polarization for Slepian-Wolf and multiple accessing using a joint decoding (and not onion-peeling) as in \cite{polar_corr,csacsouglu2011polarization,abbe2012polar}. It is conceivable that joint decoding alleviates the error propagation compared to the onion peeling approach, and thus reduces the error probability.
\appendices
\section{Proof of Lemma \ref{lem7}}\label{prooflm7}
We shall use without a proof the following basic result: 
\begin{lemma}\label{lemverdu}
Suppose $Y$ is Gaussian, and its correlation coefficient with $\tilde{X}$ is $\tilde{\rho}$, then
\begin{align}
\frac{1}{2}\log\frac{1}{1-\tilde{\rho}^2}\le I(Y;\tilde{X}).
\end{align}
\end{lemma}
Without loss of generality, assume that $X$ and $Y$ are of zero mean and unit variance. Partition the real lines with intervals $[a_{i-1},a_{i})$, $i=1,\dots k$ such that
\begin{align}
P_X([a_{i-1},a_{i}))=&\frac{1}{k},\quad i=1,\dots k,\\
a_0=&-\infty,\\
a_k=&\infty.
\end{align}
For $x\in\mathbb{R}$, let $Q(x)$ be the $[a_{i-1},a_{i})$ interval which $x$ belongs to. Define $\tilde{X}$ as a function of $X$, via
\begin{align}
\tilde{X}(x)=\mathbb{E}[X|X\in Q(x)].
\end{align}
It's easy to see that $P_{\tilde{X}}$ is equiprobable on a set of size $k$, and that $\mathbb{E}\tilde{X}^2\le \mathbb{E}X^2$. Note that
\begin{align}
\mathbb{E}[Y\tilde{X}]=&\mathbb{E}[\mathbb{E}[Y\tilde{X}|X]]\nonumber\\
=&\mathbb{E}[\rho X\tilde{X}]\nonumber\\
=&\rho\mathbb{E}[\mathbb{E}[ X\tilde{X}|\tilde{X}]]\\
=&\rho\mathbb{E}[\tilde{X}^2].
\end{align}
Hence the correlation coefficient between $\tilde{X}$ and $Y$ is
\begin{align}
\tilde{\rho}=&\frac{\rho\mathbb{E}[\tilde{X}^2]}{\sqrt{\mathbb{E}[\tilde{X}^2]}}\\
=&\rho\sqrt{\mathbb{E}[\tilde{X}^2]}.
\end{align}
Choose a number $A>0$ such that $A=a_i$ for some $i\neq k$. Integrating by parts, we have
\begin{align}\label{e26}
\int_{A}^{\infty}\frac{1}{\sqrt{2\pi}}e^{-\frac{x^2}{2}}x^2{\rm d}x
=&\frac{A}{\sqrt{2\pi}}e^{-\frac{A^2}{2}}+\frac{1}{\sqrt{2\pi}}\int_{A}^{\infty}e^{-\frac{x^2}{2}}{\rm d}x\\
\le&\frac{A}{\sqrt{2\pi}}e^{-\frac{A^2}{2}}+\frac{e^{-\frac{A^2}{2}}}{A\sqrt{2\pi}},
\end{align}
where we have used the standard bound for Gaussian cdf in the inequality. On the other hand, the length of quantization intervals in region $[-A,A]$ can be upper bounded by
\begin{align}
|a_{i-1}-a_{i}|\le& P_X(A)^{-1}\frac{1}{k}\\
=&\frac{\sqrt{2\pi}e^{\frac{A^2}{2}}}{k}.
\end{align}
This implies that
\begin{align}
\int_{-A}^{A}\frac{1}{\sqrt{2\pi}}e^{-\frac{x^2}{2}}x^2{\rm d}x-\mathbb{E}\tilde{X}^2
\le& 2\cdot\frac{\sqrt{2\pi}e^{\frac{A^2}{2}}}{k}\cdot A.
\end{align}
Combining with (\ref{e26}), we obtain
\begin{align}\label{e37}
\mathbb{E}X^2-\mathbb{E}\tilde{X}^2\le \frac{2A}{\sqrt{2\pi}}e^{-\frac{A^2}{2}}+\frac{2e^{-\frac{A^2}{2}}}{A\sqrt{2\pi}}
+2\cdot\frac{\sqrt{2\pi}e^{\frac{A^2}{2}}}{k}\cdot A
\end{align}
Now define $A_0:=\sqrt{\ln k}$ and $A:=\max\{a_i:a_i\le A_0,i=0\dots k\}$. Then (\ref{e37}) implies that
\begin{align}
&\mathbb{E}X^2-\mathbb{E}\tilde{X}^2\nonumber\\
\le& \left(\frac{2A}{\sqrt{2\pi}}+\frac{2}{A\sqrt{2\pi}}\right)e^{-\frac{A^2}{2}}
+2A\frac{\sqrt{2\pi}e^{A^2}}{k}\\
\le& \left(\frac{2A_0}{\sqrt{2\pi}}+\frac{2}{(A_0-\sqrt{\frac{2\pi}{k}})\sqrt{2\pi}}\right)e^{-\frac{(A_0-\sqrt{\frac{2\pi}{k}})^2}{2}}\nonumber\\
&+2A_0\frac{\sqrt{2\pi}e^{A_0^2}}{k}\\
=&\left(\frac{2A_0}{\sqrt{2\pi}}+\frac{2}{A_0\sqrt{2\pi}}\right)e^{-\frac{A_0^2}{2}}
+2A_0\frac{\sqrt{2\pi}e^{A_0^2}}{k}+o(\frac{1}{\sqrt{k}})\\
\le& C\sqrt{\frac{\ln k}{k}},
\end{align}
where the last step holds for any $C>\sqrt{\frac{2}{\pi}}+2\sqrt{2\pi}$ and sufficently large $k$. Using (\ref{lemverdu}), we obtain
\begin{align}
I(Y;\tilde{X})\ge&\frac{1}{2}\log\frac{1}{1-\tilde{\rho}^2}\\
\ge&\frac{1}{2}\log\frac{1}{1-\rho^2}-\frac{\log e}{2}\cdot\frac{\rho^2 C}{1-\rho^2}\sqrt{\frac{\ln k}{k}}\\
=&I(Y;X)-\frac{\log e}{2}\cdot\frac{\rho^2 C}{1-\rho^2}\sqrt{\frac{\ln k}{k}}
\end{align}
for sufficiently large $k$.

\bibliographystyle{ieeetr}
\bibliography{ref529}
\end{document}